\documentclass[12pt]{article}
\usepackage[utf8]{inputenc}

\usepackage[margin=3cm]{geometry}

\usepackage{amsmath, amssymb, amsthm}
\usepackage{color}
\usepackage{todonotes}
\usepackage{url}
\usepackage[colorlinks, allcolors=blue]{hyperref}
\usepackage{enumitem}
\usepackage{graphicx}
\usepackage{subcaption}
\usepackage{tikz}
\usepackage{cleveref}

\usepackage{authblk}

\newtheorem{theorem}{Theorem}[section]
\newtheorem{corollary}[theorem]{Corollary}
\newtheorem{lemma}[theorem]{Lemma}
\newtheorem{claim}[theorem]{Claim}

\newtheorem{proposition}[theorem]{Proposition}

\theoremstyle{definition}

\theoremstyle{remark}

\title{Treewidth is NP-Complete on Cubic Graphs  \\
(and related results)}

\author[1]{Hans~L.~Bodlaender}
\author[2]{Édouard Bonnet}
\author[3]{Lars~Jaffke}
\author[4]{Du{\v{s}}an~Knop}
\author[5]{Paloma~T.~Lima}
\author[6]{Martin~Milani{\v c}}
\author[7]{Sebastian~Ordyniak}
\author[1]{Sukanya~Pandey}
\author[4]{Ond{\v{r}}ej~Such\'{y}}

\affil[1]{Utrecht University, The Netherlands, \texttt{h.l.bodlaender@uu.nl}} 
\affil[2]{LIP, ENS Lyon, France, \texttt{edouard.bonnet@ens-lyon.fr}}
\affil[3]{University of Bergen, Norway, \texttt{lars.jaffke@uib.no}}
\affil[4]{Czech Technical University in Prague, Czech Republic, \texttt{\{dusan.knop,ondrej.suchy\}@fit.cvut.cz}}
\affil[5]{IT University of Copenhagen, Denmark, \texttt{palt@itu.dk}}
\affil[6]{FAMNIT and IAM, University of Primorska, Koper, Slovenia, \texttt{martin.milanic@upr.si}}
\affil[7]{University of Leeds, UK,  \texttt{sordyniak@gmail.com}}

\date{}

\newcommand{\tw}{\mathtt{tw}}
\newcommand{\pw}{\mathtt{pw}}
\newcommand{\cw}{\mathtt{cw}}
\newcommand{\vsn}{\mathtt{vsn}}

\begin{document}

\maketitle

\begin{abstract}
In this paper, we give a very simple proof that 
\textsc{Treewidth} is NP-complete; this proof also shows
NP-completeness on the class of co-bipartite graphs.
We then improve the result by
Bodlaender and Thilikos from 1997 that \textsc{Treewidth} is NP-complete on graphs with maximum degree at most $9$, by
showing that \textsc{Treewidth} is NP-complete on cubic graphs.
\end{abstract}

\section{Introduction}
Treewidth is one of the most studied graph parameters, with many applications for both theoretical investigations as well as for applications. The
problem of deciding the treewidth of a given graph, and finding corresponding tree
decomposition, single-handedly lead to a plethora of studies, including exact algorithms,
algorithms for special graph classes, approximations, upper and lower bound heuristics, parameterised algorithms and more. In this paper, we look at the basic problem to decide, for a given graph $G$ and integer $k$, whether the treewidth of $G$ is at most $k$.

This problem was shown to be NP-complete in 1987 by Arnborg et al.~\cite{ArnborgCP}; 
their proof also gives NP-completeness on co-bipartite graphs. As the treewidth
of a graph (without parallel edges) does not change under subdivision of edges, it
easily follows and is well known that \textsc{Treewidth}  is
NP-complete on bipartite graphs. In 1997, Bodlaender and Thilikos~\cite{BodlaenderT} modified the
construction of Arnborg et al.~and showed that \textsc{Treewidth} remains NP-complete
if we restrict the inputs to graphs with maximum degree 9. In this paper, we
sharpen this bound of 9 to 3. Our proof uses a simple transformation, whose
correctness follows from well-known facts about treewidth and simple insights.
We also give an even simpler proof of the NP-completeness of \textsc{Treewidth} 
on arbitrary (and on co-bipartite) graphs.
We obtain a number of corollaries of the results,
in particular NP-completeness of \textsc{Treewidth} on $d$-regular
graphs for each fixed $d\geq 3$, and for graphs that can be embedded
in a $3$-dimensional grid.

Our techniques are based on the techniques in \cite{ArnborgCP} and \cite{BodlaenderT}
with streamlined and simplified arguments, and some additional new but elementary
ideas. As a starting point for the reductions, we use the NP-complete problems
\textsc{Cutwidth} on cubic graphs and \textsc{Pathwidth}; the
NP-completeness proofs for these were given by Monien and Sudborough~\cite{MonienS}
in 1987.

This paper is organised as follows. In \Cref{section:definitions}, we give
basic definitions and some well-known results on treewidth.
In \Cref{section:co-bipartite}, we give a simple proof of the NP-completeness
of \textsc{Treewidth} on co-bipartite graphs that uses an elementary transformation from pathwidth. \Cref{section:main} gives our main result: NP-completeness
for \textsc{Treewidth} on cubic graphs (i.e.\  graphs with each vertex of degree 3).
In \Cref{section:corollaries}, we derive as consequences some additional NP-completeness results: on $d$-regular graphs for each fixed $d$ and on
graphs that can be embedded in a 3-dimensional grid.
Some final remarks are made in \Cref{section:conclusions}.

\section{Definitions and preliminaries}
\label{section:definitions}
Throughout the paper, we denote the number of vertices of the graph $G$ by $n$. 
All graphs considered in this paper are undirected.
A~graph $G$ is $d$-regular if each vertex has degree $d$. We say that a graph $G$ is 
\emph{cubic} if $G$ is 3-regular. If each vertex of $G$ has
degree at most 3, we say that $G$ is \emph{subcubic}.
All numbers considered are assumed to be integers, and an interval $[a,b]$ denotes the set of integers $\{a, a+1, a+2, \ldots, b-1, b\}$.
Furthermore, for a positive integer $a$, we denote by $[a]$ the interval $[1,a]$.
A~graph $G$ is a \emph{minor} of a graph $H$, if $G$ can be obtained from $H$
by zero or more vertex deletions, edge deletions, and edge contractions.
For a graph $G$ and a set of vertices $A\subseteq V(G)$, we
write $G+{\it clique}(A)$ for the graph obtained by adding an edge between
each pair of distinct non-adjacent vertices in $A$, i.e.\  by turning $A$ into
a clique.

A~\emph{tree decomposition} of a graph $G$ is a pair $(T,\beta)$ such that $T$ is a tree and $\beta$ is a~mapping assigning each node $x$ of $T$ to a \emph{bag} $\beta(x) \subseteq V(G)$, satisfying the following conditions:
every vertex of $G$ belongs to some bag,
for every edge of $G$ there exists a bag containing both endpoints of the edge, and for every vertex of $G$, the set of nodes $x$ of $T$ such that $v\in \beta(x)$ induces a connected subtree of $T$.
The \emph{width} of a tree decomposition $(T,\beta)$ is the maximum, over all nodes $x$ of $T$, of the value of $|\beta(x)|-1$.
The \emph{treewidth} of a graph $G$, denoted by $\tw(G)$, is the minimum width of a tree decomposition of $G$.
Path decompositions and pathwidth (denoted by $\pw(G)$) are defined analogously, but with the additional requirement that the tree $T$ is a path.

We use a number of well-known facts about treewidth and tree decompositions.

\begin{lemma}[Folklore]\label{lemma:basics}
Let $G$ be a graph, and $(T,\beta)$ a tree decomposition of width $k$ of $G$.
Then the following statements hold.
\begin{enumerate}
    \item Let $W$ be a clique in $G$. 
    Then, there is a node $x$ of $T$ with $W\subseteq \beta (x)$. 
    \item Suppose $v,w\in V(G)$, $\{v,w\}\not\in E(G)$. If there is a node $x$ of
    $T$, with $v,w \in \beta(x)$, then $(T,\beta)$ is a tree decomposition of
    width $k$ of the graph obtained by adding the edge $\{v,w\}$ to $G$.
    \item Suppose $W\subseteq V(G)$. 
    Then, there is a node $x$ in $T$ such that when we remove $\beta(x)$ and
    all incident edges from $G$, then each connected component of $G$ contains
    at most $n/2$ vertices of $W$.
    \item Let $y$ be a leaf of $T$, with neighbour $y'$. 
    If $\beta(y)\subseteq \beta(y')$, 
    then removing $y$ with its bag from the tree decomposition $(T,\beta)$ yields
    another tree decomposition of $G$ of width at most $k$.
    \item If $H$ is a minor of $G$, then $\tw(H)\leq \tw(G)$, and $\pw(H)\leq \pw(G)$.
\end{enumerate}  
\end{lemma}

A graph $G$ is \emph{co-bipartite} if $V(G)=A \cup B$ with $A$ a clique and $B$ a clique (that is, the complement of $G$ is bipartite). 
The following fact is also well known, and follows implicitly from
the proofs of Arnborg~et~al.~\cite{ArnborgCP}. For completeness, we give a proof here.

\begin{lemma}[See, e.g.~\cite{ArnborgCP}]
    Let $G$ be a co-bipartite graph, with  $V(G)=A \cup B$ where $A$ and $B$ are
    cliques. Then:
    \begin{enumerate}
        \item $\tw(G)=\pw(G)$.
        \item $G$ has a path decomposition $(P,\beta)$ {with width equal to $\tw(G)$} such that $A\subseteq \beta(p_1)$ and $B\subseteq \beta(p_r)$, where $p_1$ and $p_r$ are the two endpoints of $P$.
    \end{enumerate}
    \label{lemma:co-bipartite}
\end{lemma}

\begin{proof}
    Suppose $(T,\beta)$ is a tree decomposition of $G$ of width $\tw(G)$. 
    By \Cref{lemma:basics}(1), there is a node $x$ in $T$ with $A \subseteq \beta(x)$, and a node $y$ in $T$ with $B \subseteq \beta(y)$. 
    Let $P$ be the path
    from $x$ to $y$ in $T$. 

If $T$ has nodes not in $P$, then we can apply the following step. Take a leaf $z$ of~$T$, not in $P$. Let $z'$ be the neighbour of $z$ in $T$. 
For each $v\in A\cap \beta(z)$, it holds that $v\in \beta(z')$ as $z'$ is on the path from $z$ to $x$, and
for each $v\in B\cap \beta(z)$, it holds that $v\in \beta(z')$ as $z'$ is on the path from $z$ to $y$. So, by \Cref{lemma:basics}(4), we can remove $z$ from $T$ and
obtain another tree decomposition of $G$. 
Repeating this step as long as possible gives the desired result.
\end{proof}

The \emph{vertex separation number} of a graph $G$ is denoted by $\vsn(G)$ and defined as the minimum, over all orderings $\sigma = (v_1,\ldots, v_n)$ of the vertex set of $G$, of the maximum, over all $i\in \{1,\ldots, n\}$, of the number of vertices $v_j$ such that $j>i$ and $v_j$ has a neighbour in $\{v_1,\ldots, v_i\}$.
Kinnersley proved the following characterisation of pathwidth.

\begin{theorem}[Kinnersley~\cite{MR1178214}]\label{theorem:pw-equals-vsn}\label{pw-equals-vsn}
The pathwidth of every graph equals its vertex separation number.
\end{theorem}

\textsc{Treewidth} is the following decision problem: Given a graph $G$ and an integer $k$, is the treewidth of $G$ at most $k$?
The problems \textsc{Pathwidth} and \textsc{Vertex Separation Number} are defined analogously.

\begin{sloppypar}
In 1987, Arnborg, Corneil, and Proskurowski established NP-completeness of \textsc{Treewidth} in the class of co-bipartite graphs~\cite{ArnborgCP}. 
Ten years later, Bodlaender and Thilikos~\cite{BodlaenderT} proved that 
\textsc{Treewidth} is NP-complete on graphs with maximum degree at most $9$. 
Monien and Sudborough~\cite{MonienS} proved that \textsc{Vertex Separation Number} is NP-complete on planar graphs with maximum degree at most $3$.
Combining this result with \Cref{pw-equals-vsn} directly shows the following.
\end{sloppypar}

\begin{theorem}[Monien and Sudborough~\cite{MonienS}]
    \textsc{Pathwidth} is NP-complete on planar graphs with maximum degree at most $3$.
\end{theorem}

A well-known type of graphs are the \emph{brick walls}. A brick wall with $r$ rows
and $c$ columns has $r\times c$ vertices. We refrain from giving a formal definition here, as the concept is clear from \Cref{figure:brickwall}.

\begin{figure}
    \centering
    \includegraphics{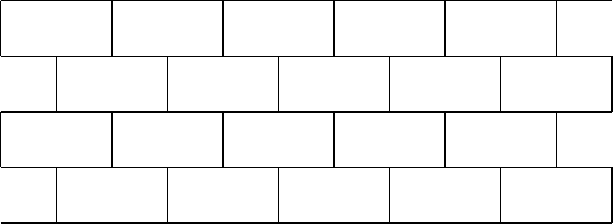}
    \caption{A brick wall with 5 rows and 12 columns.}
    \label{figure:brickwall}
\end{figure}

It is well known that the pathwidth and treewidth of an $n$ by $r$ grid
equal $\min \{n,r\}$, see, e.g.~\cite[Lemmas 87 and 88]{Bodlaender98}. 
Since any brick wall is a subgraph of a grid, the upper bound also holds for brick walls, and 
the standard construction gives the following result.

\begin{lemma}[Folklore]
    Let $B_{r,c}$ be a brick wall with $r$ rows and $c$ columns.
    Then $\tw(B_{r,c}) \leq \pw(B_{r,c}) \leq c$ and there is a path decomposition
    $(P,\beta)$ of $B_{r,c}$ of width $c$ with $\beta(p_1)$ the set of vertices
    on the first column of $B_{r,c}$, and $\beta(p_q)$ the set of vertices on the last column of $B_{r,c}$, where $p_1$ and $p_r$ are the two endpoints of $P$.
    \label{Lemma:pathdecompositionbrickwall}
\end{lemma}

A \emph{linear ordering} of a graph $G$ is a bijection $f: V(G) \rightarrow \{1, \ldots, n\}$. The \emph{cutwidth} of a linear ordering of $G$ is
\[ \max_{i \in [n]} \Big| \big\{ \{v,w\} \in E(G) ~\big|~f(v) \leq i < f(w) \big\}\Big|.\]
The \emph{cutwidth} of a graph $G$, denoted by $\cw(G)$, is the minimum cutwidth of a linear ordering of $G$.

The \textsc{Cutwidth} problem asks to decide, for a given graph $G$ 
and integer $k$, whether the cutwidth of $G$ is at most $k$. Monien and Sudborough~\cite{MonienS}
showed that \textsc{Cutwidth} is NP-complete on graphs of maximum degree three
(using the problem name \textsc{Minimum Cut Linear Arrangement}). As their proof
does not generate vertices of degree one, and the cutwidth of a graph does not change
by subdividing an edge, from their proof, the next result follows.

\begin{theorem}[Monien and Sudborough~\cite{MonienS}]
\textsc{Cutwidth} is NP-complete on cubic graphs.
\end{theorem}

\section{A simpler proof for co-bipartite graphs}
\label{section:co-bipartite}

In this section, we give a simple proof that
\textsc{Treewidth} is NP-complete. Our proof borrows elements
from the NP-completeness proof from Arnborg et al.~\cite{ArnborgCP}, but uses instead a very simple 
transformation from \textsc{Pathwidth}.

\newcommand{\TG}{F}
Let $G$ be a graph. 
We denote by $\TG(G)$ the graph obtained from~$G$ as follows.
The vertices of $\TG(G)$ consist of two copies $v$ and $v'$ for every $v \in V(G)$; we denote by $V$ and $V'$ the sets 
$V(G)$ and $\{v' \mid v \in V(G)\}$, respectively. 
Moreover, the graph $\TG(G)$ contains for every $v \in V(G)$ an edge between $v$ and $v'$, and for every edge $\{u,v\} \in E(G)$, it contains one edge between $u$ and $v'$ and one edge between $v$ and $u'$. 
Finally, $\TG(G)$ contains all edges between every pair of distinct vertices in $V$ and every pair of distinct vertices in $V'$.
Note that each of the sets $V$ and $V'$ are cliques in $\TG(G)$. 
In particular, $G$ is co-bipartite. An example is given in \Cref{figure:TG}.

\begin{figure}
    \centering
    \includegraphics[scale=1.3]{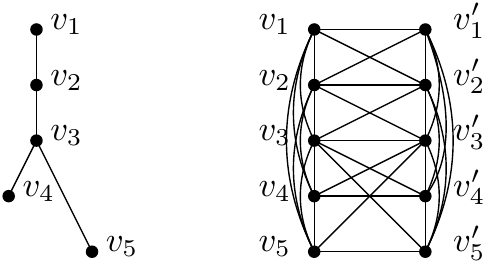}
    \caption{A graph $G$ with $\TG(G)$.}
    \label{figure:TG}
\end{figure}

\begin{lemma}\label{from-G-to-T(G)}
Let $G$ be a graph. Then, $\tw(\TG(G))= \pw(\TG(G)) = n+\pw(G)$, where $n=|V(G)|$.
\end{lemma}

\begin{proof}
First, we show that $\pw(\TG(G))\leq n+\pw(G)$. 
Let $k = \pw(G)$.
Take a path decomposition $(P,\beta)$ of $G$ of width $k$, with $P= (p_1, \ldots, p_r)$.
Now, let $\gamma(p_i)$ be a set of vertices of $\TG(G)$ defined as follows:
\begin{itemize}
    \item For each $v\in V(G)$ such that there is a $j\geq i$ with $v\in \beta(p_j)$,
    add $v$ to $\gamma(p_i)$.
    \item For each $v\in V(G)$ such that there is a $j\leq i$ with $v\in \beta(p_j)$,
    add $v'$ to $\gamma(p_i)$.
\end{itemize}

An example of this construction, applied to the graphs $G$ and $\TG(G)$ of \Cref{figure:TG}, is given in \Cref{figure:pathdec1}.

\begin{figure}
    \centering
    \includegraphics{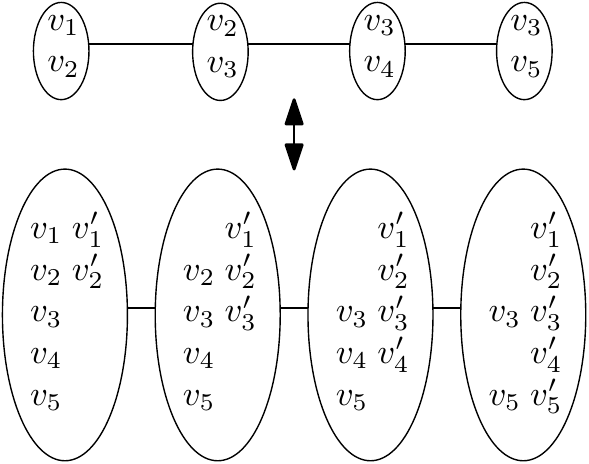}
    \caption{A path decomposition of {the graph $G$ from \Cref{figure:TG}} and the corresponding path decomposition of~$\TG(G)$.}
    \label{figure:pathdec1}
\end{figure}

We claim that $(P,\gamma)$ is a path decomposition of $\TG(G)$ of width $n+k$.
We first verify that  $(P,\gamma)$ is a path decomposition.
The first and third conditions of path decompositions are clearly satisfied. 
Notice that $V \subseteq \gamma(p_1)$, and $V'\subseteq \gamma(p_r)$. So, for each
edge in $\TG(G)$ between two vertices in 
$V$, or between two vertices in $V'$, there is a bag in  $(P,\gamma)$
containing the two endpoints of the edge, namely, the bag corresponding to the node $p_1$ or $p_r$, respectively. 
Consider an edge $\{v,v'\}$ for a vertex $v\in V(G)$. There is a node $p_v$ with
$v\in \beta(p_v)$, and therefore $v,v' \in \gamma(p_v)$.
Consider an edge $\{v,w'\}$ in $\TG(G)$, corresponding to an edge $\{v,w\}\in E(G)$. 
There is a node $p_{vw}$ with $v,w\in \beta(p_{vw})$. Now, $v, v', w, w' \in \gamma(p_{vw})$.

To see that the width is $n+k$, consider some bag $\gamma(p_i)$ and a vertex $v \in V(G)$. There are three possible cases:
\begin{enumerate}
    \item For each $j$ with $v \in \beta(p_j)$, $j>i$. Now, $v\in \gamma(p_i)$; 
    $v'\not\in \gamma(p_i)$.
    \item For each $j$ with $v \in \beta(p_j)$, $j<i$. Now, $v'\in \gamma(p_i)$; 
    $v\not\in \gamma(p_i)$.
    \item If the previous two cases do not hold, then there is $j\leq i$ with
    $v \in \beta(p_j)$, and $j'\geq i$ with $v\in \beta(p_{j'})$. From the
    definition of path decompositions, it follows that $v\in \beta(p_i)$. 
    From the construction of $\gamma$, we have $v,v'\in \gamma(p_i)$.
\end{enumerate}

In each of the cases, we have one vertex more in $\gamma(p_i)$ than in $\beta(p_i)$,
so for each node, the size of its $\gamma$-bag is exactly $n$ larger than the size
of its $\beta$-bag. The claim follows.

\medskip

Now, suppose the treewidth of $G$ equals $\ell$. 
From \Cref{lemma:co-bipartite}(2),
it follows that
we can assume we have a path decomposition $(P,\gamma)$ of $\TG(G)$ of
width $\ell$, with $P$ having successive bags $p_1, p_2, \ldots, p_r$,
and with $V \subseteq \gamma(p_1)$ and $V'\subseteq \gamma(p_2)$.

We now define a path decomposition $(P,\delta)$ of $G$, as follows.
For each node $x$ on~$P$, set $\delta(x) = \{v\in V ~|~ v\in \gamma(x) \wedge 
v'\in \gamma(x)\}$. (Note that this is the reverse of the operation in the first
part of the proof; compare with \Cref{figure:pathdec1}.)

We now verify that $(P,\delta)$ is indeed a path decomposition of $G$. For
each vertex $v$, $\{v,v'\}$ is an edge in $\TG(G)$, so there is a node $x_v$ with
$v,v'\in \gamma(x_v)$, hence $v\in \delta(x_v)$. For each edge $\{v,w\}\in E(G)$,
the set $\{v,v',w,w'\}$ forms a clique in $\TG(G)$, so there is a node
$x_{vw}$ with  $\{v,v',w,w'\} \subseteq \gamma(x_{vw})$ (see \Cref{lemma:basics}(1)).
Hence $v,w\in \delta(x_{vw})$. Finally, for each $v\in V(G)$, the set of
nodes $x$ with $v\in \delta(x)$ is the intersection of the nodes with
$v\in \gamma(x)$ and the nodes with $v'\in \gamma(x)$; the intersection of 
connected subtrees is connected, so the third condition in the definition of path (tree) decompositions also holds.

Finally, we show that the width of $(P,\delta)$ is $\ell-n$.
Consider a vertex $v$, and $i\in [r]$.
There must be $i_v$
with $\{v,v'\} \subseteq \gamma(p_{i_v})$. If $i\leq i_v$, then $v\in \gamma(p_i)$;
if $i\geq i_v$, then $v'\in \gamma(p_i)$ (using that $v\in \gamma(p_1)$ and
$v'\in \gamma(p_r)$). So, we have $\{v,v'\}\cap \gamma(p_i) \neq \emptyset$. 

Now, for each node $p_i$, $i\in [r]$, for each vertex $v$, we have
that $\gamma(p_i)$ contains both vertices from the set 
$\{v,v'\}$
when $v\in \delta(p_i)$, and $\gamma(p_i)$ contains exactly one vertex from the set $\{v,v'\}$ 
when $v\not\in \delta(p_i)$. So, $|\gamma(p_i)| = | \delta(p_i)|+n$. As this holds
for each bag, we have that the width of $(P,\gamma)$ is exactly $n$ larger than the
width of $(P,\delta)$. It follows that $\pw(G) \leq \tw(\TG(G))-n \leq \pw(\TG(G))-n$, which shows the result.
\end{proof}

\Cref{from-G-to-T(G)}, together with the NP-completeness of \textsc{Vertex Separation Number}~\cite{MonienS}, and the equivalence between the pathwidth and the vertex separation number (\Cref{pw-equals-vsn}), leads to an alternative and simpler
proof of NP-completeness of \textsc{Treewidth} in the class of co-bipartite graphs.

\begin{corollary}
\textsc{Treewidth} is NP-complete on co-bipartite graphs.
\end{corollary}

One can obtain a proof of the NP-completeness of 
\textsc{Treewidth} on graphs with maximum degree five
by combining the proof above with the technique of
replacing a~grid with a~brick wall or grid (as in \cite{BodlaenderT} or in the next section).
Instead of this, we give in the next section a proof
that reduces from \textsc{Cutwidth} and shows NP-completeness of \textsc{Treewidth} on graphs of degree three.

\section{Cubic graphs}
\label{section:main}
In this section, we give an NP-completeness proof
for \textsc{Treewidth} on cubic graphs. The 
construction uses a few steps. The first step
is a simplified version of the NP-completeness 
proof from Arnborg et al.~\cite{ArnborgCP}; the
second step follows the idea of Bodlaender and
Thilikos~\cite{BodlaenderT} to replace the cliques by grids or
brick walls. After this step, we have a graph
with maximum degree 7. In the third step, we
replace vertices of degree more than 3 by trees
of maximum degree 3, and show that this step does
not change the treewidth (it actually can change
the pathwidth). The fourth step makes the graph 
3-regular by simply contracting over vertices of degree 2.

\begin{theorem}\label{thm:cubic}
\textsc{Treewidth} is NP-complete on regular graphs of degree 3.
\end{theorem}

\begin{proof}
We use a transformation from \textsc{Cutwidth} on 3-regular graphs.

Let $G$ be an $n$-vertex $3$-regular graph and $k$ an integer.
Using a sequence of intermediate steps and intermediate graphs $G_1$, $G_2$, $G_3$, 
we construct a 3-regular graph $G_4$ with the property that $G$ has cutwidth at most $k$, if and only if $G_4$ has treewidth at most $3n+k+2$.

\paragraph{Step 1: From Cutwidth to Treewidth}
The first step is a streamlined version of the proof from Arnborg et al.~\cite{ArnborgCP}. For each vertex $v\in V(G)$, we take a set $A_v = \{v^1, v^2, v^3\}$ which has three copies of $v$.

For each edge $e\in E(G)$, we have a set $B_e = \{e^1,e^2\}$, which consists of two vertices that represent the edge.

Let $A= \bigcup_{v\in V(G)}A_v$, and $B=\bigcup_{e\in E(G)} B_e$.
We create $G_1$ by taking $A\cup B$ as vertex set, turning $A$ into a clique, 
turning $B$ into a clique, and for each pair $v$, 
$e$ with $v$ an endpoint of $e$, adding edges between all vertices in $A_v$ and
all vertices in $B_e$.

\begin{claim}
    Let $G$ and $G_1$ be as above. $\tw(G_1)=\pw(G_1)=
    \cw(G)+3n+2$.
    \label{claim:g1}
\end{claim}

\begin{proof}
First, assume $G$ has cutwidth $k$, and let
$f$ be a linear ordering of $G$ of
cutwidth~$k$, and denote the $i$th vertex in
the linear ordering as $v_i=f^{-1}(i)$.

Build a path decomposition $(P,\beta)$ 
with $P$ the path with nodes $p_1$, \ldots, $p_n$.
For $i\in [n]$, set 
\begin{equation*}
\begin{split}
\beta(p_i) = & \left\{ v_j^a ~\big|~ j\geq i  \wedge a\in
\{1,2,3\}\right\} \\ 
& \cup \left\{ e^b ~\big|~ e=\{v_{j},v_{j'}\}\in E(G) \wedge \min\{j,j'\} \leq i \wedge b\in [2] \right\}.
\end{split}
\end{equation*}
That is, we take the representatives of the vertices
$v_i, v_{i+1}, \ldots, v_n$, and all vertices
that represent an edge with at least one endpoint
in $\{v_1, v_2, \ldots, v_i\}$.

We can verify that $(P,\beta)$ is a path decomposition of $G_1$.
From the construction, it directly follows that $A \subseteq \beta(p_1)$
and 
$B \subseteq \beta(p_n)$. For the second condition of path decompositions, it remains to look at edges in $G_1$
with one vertex of the form $v_i^a$ and one 
vertex of the form $e^b$. Necessarily, $v_i$ is
an endpoint of $e$, and now we can note that both vertices
are in bag $\beta(p_i)$.
From the construction, it directly
follows that the third condition of path decompositions
is fulfilled. 

To show that the width of this path decomposition
is at most $k+3n+2$, we use an accounting system.
Consider $\beta(p_i)$. Give each vertex $v\in V(G)$
three credits, except $v_i$, which gets six credits.
Each edge that `crosses the cut', i.e.\  it belongs to the set
$\{\{v,w\}\in E(G) ~|~ f(v) \leq i < f(w)\}$, gets one credit. All other edges get no credit.
We handed out at most $k+3n+3$ credits.
{We now redistribute these credits to the vertices in $\beta(p_i)$.}
Each vertex $v_j$, $j\geq i$, gives one credit
to each vertex of the form $v_j^a$, $a\in \{1,2,3\}$.
For an edge $e=\{v_j,v_{j'}\}$, with $j< i$ and $j'<i$,
the vertices $e^1$ and $e^2$ get, respectively,
a credit from $v_j$ and $v_{j'}$.
For an edge $e=\{v_j,v_{j'}\}$, with $j\leq i < j'$,
the vertices $e^1$ and $e^2$ get, respectively,
a credit from $v_j$ and a credit from $e$.
Now, each vertex and edge precisely spends its
credit: a vertex $v_j$ with $j<i$ gives one credit
to each of its incident edges, $v_i$ gives one credit
to each of its copies $v_i^1$, $v_i^2$, $v_i^3$, and one credit to each of its incident edges,
and $v_j$ with $j>i$ gives one credit to each
of its copies $v_j^1$, $v_j^2$, $v_j^3$.
Each vertex in the bag $\beta(p_i)$ gets one
credit, so the size of the bag is at most
$k+3n+3$. As this holds for each bag, the width
of the path decomposition is at most $k+3n+2$.

\medskip

Now, assume that we have a tree decomposition
$(T,\gamma)$ of $G_1$ of width $\ell$. By \Cref{lemma:basics}(1),
as $A$ and $B$ are cliques,
there is a bag $p_1$ with $A\subseteq \gamma(p_1)$,
and a bag $p_r$ with $B \subseteq \gamma(p_r)$.
{As in the proof of \Cref{lemma:co-bipartite},}
we can remove all bags not on the path from $p_1$ and $p_r$,
and still keep a tree decomposition of $G_1$. 
So, we can assume we have a path decomposition $(P,\gamma)$ of width {at most} $\ell$ of $G_1$, where $P$ is a path with successive vertices $p_1, p_2, \ldots, p_r$, and $\gamma(p_1)=A$ and
$\gamma(p_r) = B$.

For each $v\in V(G)$, set $g(v)$ to the maximum $i$ such
that $\{v^1,v^2,v^3\} \subseteq \beta(p_i)$.
(As $\{v^1,v^2,v^3\} \subseteq A \subseteq \beta(p_1)$, $g(v)$ is well defined
and in $[r]$.)

Take a linear ordering $f$ of $G$ such that for all $v,w\in V(G)$, $g(v)<g(w) \Rightarrow f(v)<f(w)$. (That is, order the vertices with respect to increasing values of $g$, and
arbitrarily break the ties when vertices have the same value $g(v)$.) We claim that
$f$ has cutwidth at most $\ell - 3n - 2$.

Consider a vertex $v\in V(G)$, and suppose $g(v)=i'$. Let $e$ be an edge incident to~$v$.
The set $\{v^1, v^2, v^3, e^1, e^2\}$ is a clique in $G_1$, so there is an $i_e$ with
$\{v^1, v^2, v^3, e^1, e^2\} \subseteq \beta(p_{i_e})$. From the definition of path
decompositions and the construction of $g$, we have $i_e\leq i'$. 
As $\{e^1, e^2\}\subseteq \beta(p_{i_e})\cap \beta(p_r)$, we have that $\{e^1, e^2\}\subseteq \beta(p_{i'})$.

Now, consider an $i\in [n]$. Let $v= f^{-1}(i)$ be the $i$th vertex of the ordering and $C = f^{-1}[i]$ be the first $i$ vertices in the linear ordering.
Let $E^1$ be the set of edges with exactly one endpoint in $C$, and let $E^2$ be the set
of edges with both endpoints in $C$.
Suppose $g(v)=i'$.
We now examine which vertices belong to $\beta(p_{i'})$:
\begin{itemize}
    \item By definition, $v^1$, $v^2$, $v^3$.
    \item For each $w\in V(G)\setminus C$, there is an $i_w\geq i'$ with 
    $\{w^1,w^2,w^3\}\subseteq \beta(p_{i_w})$, hence $w^1$, $w^2$, and $w^3$ are in
    $\beta(p_{i'})$. (Use here that these vertices are in $\beta(p_1)$.)
    The number of such vertices is $3n-3i$.
    \item For each edge $e\in E^1 \cup E^2$, from the discussion above
    it follows that there is an $i_e \leq i'$ with $e^1, e^2 \in \beta(p_{i_e})$,
    and, as these vertices are in $\beta(p_r)$, we have $\{e^1,e^2\} \subseteq \beta(p_{i'})$. 
\end{itemize}
Thus, the size of $\beta(p_{i'})$ is at least $3n-3i+3 + 2\cdot|E_1|+2\cdot|E_2|$.
As each vertex in~$C$ is incident to exactly three edges, we have $3i = |E_1| + 2 \cdot |E_2|$.
Now, $\ell \geq |\beta(p_{i'})| -1 \geq 3n-3i+2 + 2\cdot|E_1|+2\cdot|E_2| = 3n+2 + |E_1|$.
It follows that the size of the cut $\Big|\big\{\{x,y\}\in E(G) ~\big|~ f(x)\leq i < f(y)\big\}\Big| 
= |E_1| \leq \ell -3n -2$. As this holds for each $i\in [n]$, the bound of $\ell-3n-2$
on the cutwidth of $f$ follows.

\medskip
{We have thus shown that $\pw(G_1)\le
    \cw(G)+3n+2$ and that $\cw(G_1)\le \tw(G_1)-3n-2$.
    Together with the inequality $\tw(G_1)\le \pw(G_1)$, this proves the claim.}
\end{proof}

\paragraph{Step 2: The brick wall construction}
In the second step, we use a technique from Bodlaender and Thilikos~\cite{BodlaenderT}.
We construct a graph $G_2$ from the graph $G_1$ by removing the edges between vertices in $A$ and the edges between vertices in $B$; then, we add a brick wall with $3n$ rows
and $24n$ columns, and add a matching from the vertices in the last column of the brick wall to the vertices in $A$. 
Similarly, we add another brick wall with $3n$ rows and $24n$ columns, and add a matching from the vertices in the first column of this brick wall to the vertices in $B$.

As applying the brick wall construction to a graph obtained 
from the first step would be unwieldy, the 
example in \Cref{figure:brickwallconstruction} shows
the brick wall construction applied to the graph from
the previous section.

\begin{figure}
    \centering
    \includegraphics{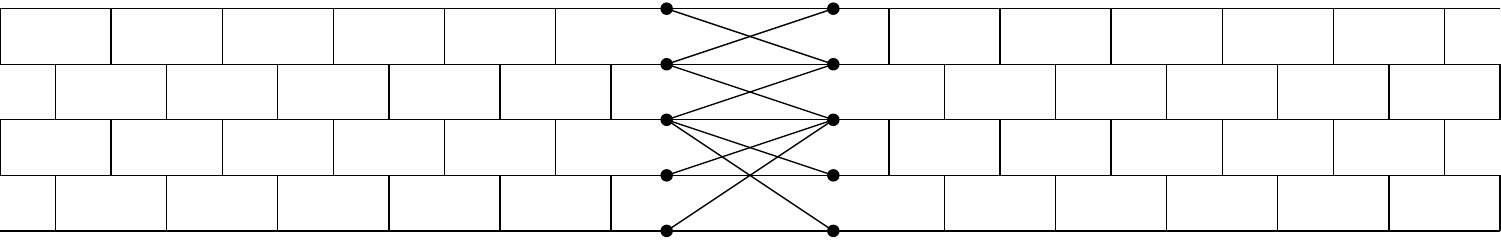}
    \caption{Illustration of the the brick wall construction.
Here, it is applied to the graphs from \Cref{figure:TG},
and the number of columns shown is smaller than
that in the actual construction.}
    \label{figure:brickwallconstruction}
\end{figure}

\begin{claim}\label{claim:tw:G1:G2}
    $\tw(G_1)=\pw(G_1)=\tw(G_2)=\pw(G_2)$. Moreover, there is a path decomposi\-tion
    of $G_2$ of optimal width with a node $x_A$ with $A\subseteq \beta(x_A)$ and
    a node $x_B$ with $B\subseteq \beta(x_B)$.
\end{claim}

\begin{proof}
Suppose we have a tree decomposition $(T,\beta)$ of $G_2$ of {optimal} width $k$.
By \Cref{lemma:basics}(3), there is a node $x$ such that each connected component
of $G_2 \setminus \beta(x)$ contains at most $36n^2$ vertices of the left brick wall.
Note that $\beta(x)$ must contain a vertex of each row from the left brick wall.
{Suppose not. Each pair of two successive columns in the brick wall is connected;
there are at least $12n- |\beta(x)|$ disjoint pairs of columns which do not contain a vertex from $\beta(x)$. 
All vertices on these columns are connected in $G_2 \setminus \beta(x)$ as they intersect the row without vertices in $\beta(x)$. 
As the number of vertices in these columns is larger than $36n^2$, since $k \leq |E(G)| = 3n/2$, we have a contradiction.
}

By \Cref{lemma:basics}(2), $(T,\beta)$ is also a tree decomposition of the graph
obtained from $G_2$ by adding edges between each pair of vertices in $\beta(x)$.
Apply the same step to the right brick wall. 
We see that $(T,\beta)$ is a tree decomposition of width $k$ of a graph that for each pair of rows in the left brick wall contains an edge between a pair of vertices from these rows, and similarly for the right brick wall. 
Now, if we contract each row of the left brick wall to the neighbouring vertex in $A$, and contract each row of the right brick wall to the neighbouring vertex in $B$, we obtain $G_1$ as minor: $G_1$ is a minor of a graph of treewidth $k$, so has
treewidth at most $k$.

\medskip

By \Cref{lemma:co-bipartite}, $\tw(G_1)=\pw(G_1)$, and there is a
path decomposition $(P,\gamma)$ of $G_1$ of {optimal} width $\ell$ such that $A\subseteq \gamma(p_1)$ and $B\subseteq \gamma(p_q)$, where $p_1$ and $p_q$ are the endpoints~of~$P$.

We can now build a path decomposition of $G_2$ {of the same width $\ell$} as follows: first, take the
successive bags of a path decomposition of the left brick wall, of width $3n$, where
we can end with a bag that contains all vertices of $A$. Then, we take the bags
of $(P,\gamma)$. Now, we add a path decomposition of the right brick wall, of width $3n$, that starts with a bag containing all vertices in $B$.
\end{proof}

\paragraph{Step 3: Making the graph subcubic}
Note that the maximum degree of a vertex in $G_2$ is seven. A vertex in $A$ has
one neighbour in the brick wall, and six neighbours in $B$ (the vertex it represents
has three incident edges,
and each is represented by two vertices). Similarly, a vertex in $B$ has degree
seven: again, one neighbour in the brick wall,
and six neighbours in $A$ (each endpoint of
the edge it represents is represented by three vertices). Vertices in the brick walls have
degree {at most} three.

Given $G_2$, we build a subcubic graph $G_3$.
We do this by replacing each vertex in $A$ and in $B$ by a tree, and replacing edges to vertices in $A$ and $B$ by edges to leaves or the root of these trees.

For vertices $v^\alpha$ in $A$ (with $v\in V(G)$, $\alpha\in [3]$), we take an arbitrary tree with a root of degree 2, all
other internal vertices of degree 3, and six leaves.
The root (which we denote by the name of the original vertex $v^\alpha$) is made adjacent to the neighbour of $v^\alpha$ in the brick wall.

Each vertex $e^\alpha \in B$ (with $e \in E(G)$, $\alpha\in [2]$) is also replaced
by a tree with a root of degree 2, all other internal vertices of degree 3, and six
leaves, but here we need to use a specific shape of the tree. 
Suppose $e$ has  endpoints $v$ and $w$. \Cref{figure:constructionB} shows this tree. In
particular, note that the root is made adjacent to the neighbour of $e^\alpha$ in the brick
wall, and the leaves that go to the subtrees that represent $v$ are grouped together,
and the leaves that go to the subtrees that represent $w$ are grouped together.

\begin{figure}
    \centering
    \includegraphics{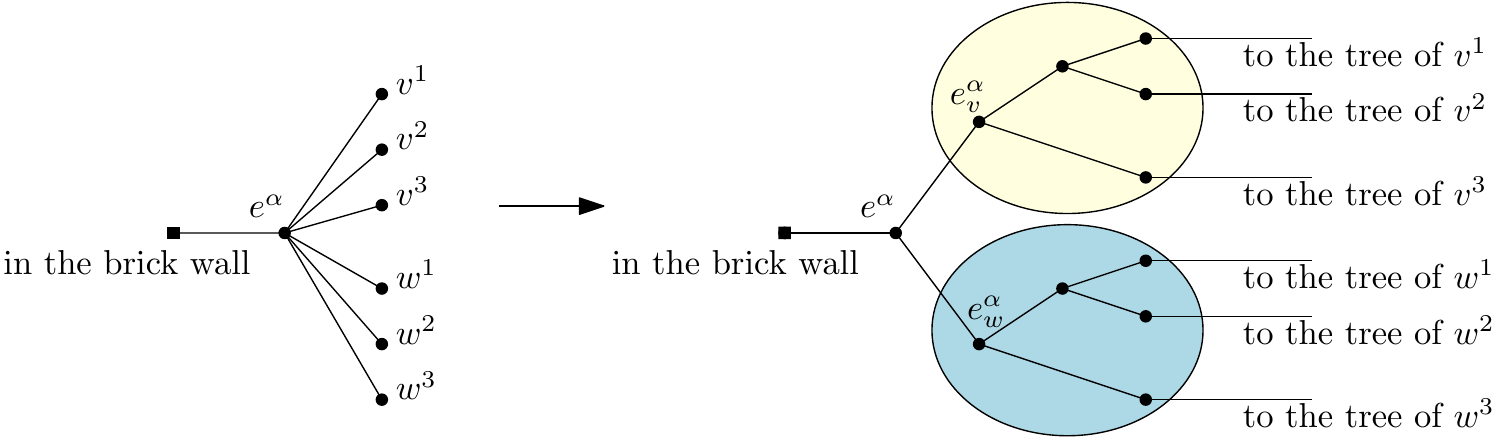}
    \caption{Replacing a vertex $e^\alpha$ from $B$ by a tree; $e$ is here the edge $\{v,w\}$.}
    \label{figure:constructionB}
\end{figure}

Each edge between a vertex $v^\alpha$ in $A$ and a vertex  $e^{\alpha'}$in $B$ now becomes an edge from a
leaf of the tree representing $v^\alpha$, to a leaf of the tree representing $e^{\alpha'}$; $\alpha\in [3]$, $\alpha'\in [2]$.
The roots of the trees are made adjacent to a vertex in the brick wall; this is the
same vertex as the brick wall neighbour of the original vertex in $G_2$.

\begin{claim}
    Suppose $\tw(G_2) \geq 68$. Then
    $\tw(G_2)=\pw(G_2)=\tw(G_3)$.
\end{claim}

\begin{proof}
    We have already established that $\tw(G_2)=\pw(G_2)$.

    First, note that $G_2$ is a minor of $G_3$: we obtain $G_2$ from $G_3$
    by contracting each of the new trees to its original vertex. By \Cref{lemma:basics}(5), we have $\tw(G_2) \leq \tw(G_3)$.

    \medskip

    Suppose we have a path decomposition $(P,\beta)$ of $G_2$ of optimal
    width $\ell = \pw(G_2)=\tw(G_2)$. 
    {By \Cref{claim:tw:G1:G2}}, we can also assume that there is a bag that contains all vertices in $A$, and that there is a bag that
    contains all vertices in $B$.

    For each vertex $v\in V(G)$, we claim that there is a node $p_{i_v}$ with
    $v^1, v^2, v^3 \in \beta(p_{i_v})$ and $e^1, e^2 \in \beta(p_{i_v})$ for each edge $e$ incident to $v$. 
    This can be shown as follows. 
    The pair $(P,\beta)$ is also a path decomposition of the graph $G+{\it clique}(A)+{\it clique}(B)$, obtained from $G_2$ by adding edges between each pair of vertices in $A$, and each pair of
    vertices in $B$ (since there is a bag containing all vertices of $A$ and a bag containing all vertices of $B$ and by
    \Cref{lemma:basics}(2).) The claim now follows from \Cref{lemma:basics}(1)
    by observing that these nine vertices ($v^1, v^2, v^3$, and $e^1$, $e^2$ for each
    edge incident to $v$) form a clique in $G+{\it clique}(A)+{\it clique}(B)$.

   Now, we can construct a tree decomposition of $G_3$ as follows.
    Take $(P,\beta)$. Replace each vertex in $A$ and each vertex in $B$ by the root
    of the tree it represents. For each vertex $v\in V(G)$, we add one additional
    bag to the tree decomposition; this bag becomes a leaf of the tree decomposition.
    (Note that after this step, we no longer have a path decomposition.)

    Consider a vertex $v\in V(G)$. Take a new node $x_v$, and make $x_v$ adjacent
    to $p_{i_v}$ in the tree. Let the bag of $x_v$ contain the following vertices:
    all vertices in the subtrees that represent $v^1$, $v^2$, $v^3$,
    for each edge $e$ with $v$ as endpoint the vertices $e^1$, $e^1_v$, $e^2$,
    $e^2_v$, and the descendants of $e^1_v$ and $e^2_v$ in the respective subtrees
    (the vertices in the yellow area in \Cref{figure:constructionB}, 
    assuming that $e = \{v,w\}$).

    \begin{figure}
        \centering
        \includegraphics[width=\textwidth]{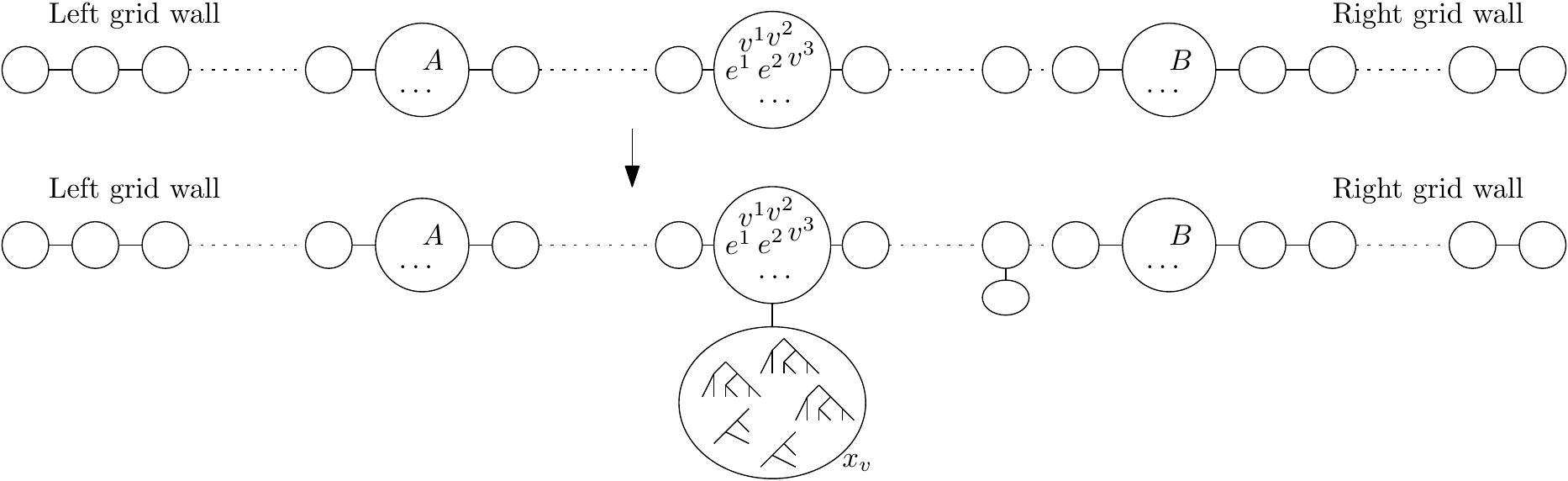}
        \caption{Illustration of the proof. The path decomposition before and after adding the new node $x_v$.}
        \label{figure:step3}
    \end{figure}

    Each vertex in $A$ is represented by a binary tree with a root of degree two
    and six leaves, so by eleven vertices. For each of the three edges incident to $v$, we
    have two subtrees of which we take six vertices each, so the total size
    of this new bag is $3\cdot 11+ 3\cdot 2 \cdot 6 = 69$.
    One easily verifies that we have a tree decomposition of $G_3$, and as the
    original bags keep the same size when $\ell \geq 68$, we have a tree decomposition of $G$ of width at most $\ell$.
\end{proof}

By taking a sufficiently large $n$ (e.g.\ $n\geq 22$ works), we can assume
that $\ell \geq 68$.

\paragraph{Step 4: Making the graph 3-regular}
The fourth step is simple. Note that when the treewidth of a graph is at least three, the treewidth does not change when we contract a vertex of degree at most two to a neighbour (see \cite{ArnborgProskurowski86}), possibly removing parallel edges. 
We apply this step as long as possible, and let $G_4$ be the resulting graph.
The graph $G_4$ is a 3-regular graph, and, when $n\geq 22$, its treewidth
equals the treewidth of $G_1$, which is $\cw(G)+3n+2$. As we can construct
$G_4$ in polynomial time, this completes the transformation, and we can
conclude that \textsc{Treewidth} is NP-complete on 3-regular graphs.
\end{proof}

\section{Special cases}
\label{section:corollaries}
In this section, we give two  NP-completeness proofs
for \textsc{Treewidth} on special graph classes, which
follow from minor modifications of the proof of \Cref{thm:cubic}.
We first observe that for any fixed $d\geq 4$, \textsc{Treewidth} is NP-complete
on $d$-regular graphs.

\begin{proposition}
    For each $d \geq 3$, \textsc{Treewidth} is NP-complete on $d$-regular graphs.
\end{proposition}

\begin{proof}
The result for $d=3$ was given as \Cref{thm:cubic}.

    A small modification of the proof of \Cref{thm:cubic} gives the result for \emph{$4$-regular} graphs: instead of using a brick wall, use a grid. At the borders of this grid, we have vertices of degree less than 3. 
    We can avoid these by first contracting vertices of degree 2, and then noting that there is a perfect matching with the vertices of degree 3 at the sides
    of the grid. Replace each edge in this matching by a small subgraph, as shown
    in Figure~\ref{figure:degree4}. Note that this step increases the
    degree of $v$ and $w$ by one, while, when the treewidth of $G$ is at least
    5, the step will not change the treewidth of the graph.

    \begin{figure}
        \centering
        \includegraphics{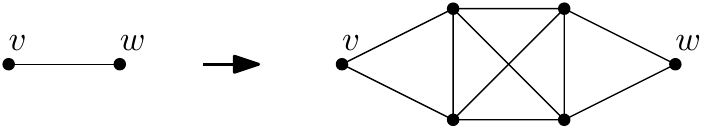}
        \caption{Increasing the degree of two adjacent vertices by one.}
        \label{figure:degree4}
    \end{figure}
    In the step where we change vertices of degree 7 to vertices of degree 3 by replacing a vertex by a small tree, we instead use a tree with the root having two children, each with three
    children. These roots are made adjacent to the grid.
    Now, the roots have degree 3, and we add an arbitrary perfect
    matching between these root vertices in $A$, and similarly for $B$. (Note that in the
    construction, there is a bag containing all roots for $A$,
    and similarly $B$; these sets have even size.)
    This gives the result for $d=4$.

    Consider the following gadget. 
    Take a clique with $d+1$ vertices, and remove one edge, say $\{x,y\}$, from this clique. For a vertex $v$ in a graph $G$, add an edge from $x$ to $v$, and an edge from $y$ to $v$. See \Cref{figure:increasedegree}.     
    \begin{figure}
    \centering
    \includegraphics{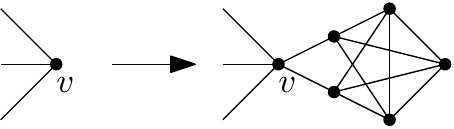}
    \caption{Increasing the degree of a vertex: if $\tw(G)\geq 4$, then the
    step increases the degree of $v$ from 3 to 5, but does not change the treewidth.}
    \label{figure:increasedegree}
\end{figure}

    If $G$ has treewidth at least $d$, then this step increases the degree of $v$ by $2$ without changing the treewidth.
    Now, if $d$ is odd, we can take an instance of the hardness proof on 3-regular graphs, and add to each vertex of that instance $(d-3)/2$ copies of this gadget. We obtain an equivalent instance that is $d$-regular.
    If $d$ is even, we add $(d-4)/2$ copies of the gadget to an instance of the hardness proof on $4$-regular graphs.
\end{proof}

A~\emph{$d$-dimensional grid graph} is a~finite induced subgraph of the infinite $d$-dimensional grid.
Observe that $d$-dimensional grid graphs have degree at most $2d$, and in particular the 3-dimensional grid graphs have degree at most 6. 
As a~consequence of lowering the degree of hard \textsc{Treewidth} instances from 9 to at most 6, we can show that computing the treewidth of 3-dimensional grid graphs is NP-complete.
Since we lowered the degree of hard instances down to at most 3, we can even show the following.

\begin{proposition}\label{prop:3d-grid}
\textsc{Treewidth} is NP-complete on subcubic 3-dimensional grid graphs.
\end{proposition}
\begin{proof}
The argument is simply that every $n$-vertex (sub)cubic graph admits a~subdivision of polynomial size that is a 3-dimensional grid graph.
We give a simple such embedding.

We reduce from \textsc{Treewidth} on cubic graphs, which is NP-hard by \Cref{thm:cubic}.
Let $G$ be any cubic graph, $v_0, v_1, \ldots, v_{n-1}$ its vertices, and $e_1, e_2, \ldots, e_{3n/2}$ its edges.
We build a~subcubic induced subgraph $H$ of the $(6n-1) \times (3n+1) \times 3$ grid that is a~subdivision of $G$.
In particular, $\tw(H)=\tw(G)$ and $H$ has $O(n^2)$ vertices and edges, thus we can conclude.

For each $i \in [0,n-1]$, vertex $v_i$ is encoded by the path made by the 5 vertices $(x,0,0)$ with $x \in [6i,6i+4]$.
We arbitrarily assign $(6i,0,0)$, $(6i+2,0,0)$, $(6i+4,0,0)$ each with a~distinct neighbour of $v_i$ in $G$, say $v_{i(0)}$, $v_{i(1)}$, $v_{i(2)}$, respectively.

Every edge $e_k=\{v_i,v_j\}$ of $G$ with $i < j$ is encoded in the following way.
Let $a,b \in [0,2]$ be such that $i(a)=j$ and $j(b)=i$.
We build a path from $(6i+2a,0,0)$ to $(6j+2b,0,0)$ with degree-2 vertices, by first adding all the vertices $(6i+2a,y,0)$ and $(6j+2b,y,0)$ for $y \in [2k]$, 
then bridging $(6i+2a,2k,0)$ and $(6j+2b,2k,0)$ by adding $(6i+2a,2k,1)(6i+2a,2k,2)(6i+2a+1,2k,2)(6i+2a+2,2k,2) \ldots (6j+2b-1,2k,2)(6j+2b,2k,2)(6j+2b,2k,1)$.

This finishes the construction of $H$.
All of its vertices have degree 2, except the vertices at $(6i+2,0,0)$, which have degree 3.
It is easy to see that $H$ is a subdivision of $G$ (where each edge gets subdivided at most $12n+5$ times).
\end{proof}

We can easily adapt the previous proof to show hardness for finite subcubic (non-induced) subgraphs of the $\infty \times \infty \times 2$ grid.
\section{Conclusions}
\label{section:conclusions}

In this paper, we gave a number of NP-completeness proofs for \textsc{Treewidth}.
The first proof is an elementary reduction from \textsc{Pathwidth} to 
\textsc{Treewidth} on co-bipartite graphs; while the hardness result is long known,
our new proof has the advantage of being very simple, and presentable in a matter of
minutes. 
Our second main result is the NP-completeness proof for
\textsc{Treewidth} on cubic graphs, which improves upon the
over 25-years-old bound of degree 9.

We end this paper with a few open problems. A long standing open problem is the
complexity of \textsc{Treewidth} on planar graphs. While the famous ratcatcher
algorithm solves the related \textsc{Branchwidth} problem in polynomial time
\cite{SeymourT94}, it is still unknown whether \textsc{Treewidth} on planar graphs
is polynomial time solvable or whether it is NP-complete. Also, no NP-hardness proofs
for \textsc{Treewidth} on graphs of bounded genus, or $H$-minor free graphs for some
fixed $H$ are known. An easier open problem might be the complexity of
\textsc{Branchwidth} for graphs of bounded degree, and we conjecture that
\textsc{Branchwidth} is NP-complete on cubic graphs.

While `our' reductions are simple, the NP-hardness of \textsc{Treewidth} is derived
from the NP-hardness of \textsc{Pathwidth} or \textsc{Cutwidth}. Thus, it
would be good to have simple NP-hardness proofs for \textsc{Pathwidth} and/or
\textsc{Cutwidth}, preferably building upon `classic' NP-hard problems like
\textsc{Satisfiability}, elementary graph problems like \textsc{Clique}, 
or \textsc{Bin Packing}.

The reductions in our hardness proofs increase the parameter by a term linear in $n$,
so shed no light on the parameterised complexity of \textsc{Treewidth}. Hence, it would
be interesting to obtain parameterised reductions (i.e.\  reductions that
change $k$ to a value bounded by a function of $k$), and also aim at
lower bounds  (e.g.\ based on the (S)ETH) on the parameterised complexity of \textsc{Treewidth}.

\paragraph{Acknowledgements.}  
This research was conducted in the Lorentz Center, Leiden, the Netherlands, 
during the workshop {\it Graph Decompositions: Small Width, Big Challenges}, October 24 -- 28, 2022. 
Martin~Milani{\v c} acknowledges the support of
the Slovenian Research Agency (I0-0035, research program P1-0285 and research projects N1-0102, N1-0160, J1-3001, J1-3002, J1-3003 and J1-4008).
Dušan Knop and Ondřej Suchý acknowledge the support of the OP VVV MEYS funded project CZ.02.1.01/0.0/0.0/16\_019/0000765 ``Research Center for Informatics''.
\bibliographystyle{abbrvurl}
\bibliography{biblio}
\end{document}